\documentclass[11pt]{article}

\usepackage[margin=1.1in]{geometry}
\usepackage[dvipsnames]{xcolor}
\usepackage{amsmath,amssymb,amsthm,mathtools}
\usepackage{microtype}
\usepackage{booktabs,array}
\usepackage{enumitem}
\usepackage{authblk}
\usepackage{algorithm}
\usepackage{algpseudocode}
\usepackage{tikz}
\usetikzlibrary{arrows.meta,fit}
\usepackage[colorlinks=true,linkcolor=blue!50!black,citecolor=green!35!black,urlcolor=blue!50!black]{hyperref}

\usepackage[color=Thistle]{todonotes}
\setlist{itemsep=2pt,topsep=4pt}
\newtheorem{theorem}{Theorem}[section]
\newtheorem{lemma}[theorem]{Lemma}
\newtheorem{proposition}[theorem]{Proposition}
\newtheorem{corollary}[theorem]{Corollary}
\newtheorem{fact}[theorem]{Fact}

\theoremstyle{definition}
\newtheorem{definition}[theorem]{Definition}

\newtheorem{example}[theorem]{Example}
\theoremstyle{remark}

\newcommand{\polylog}{\operatorname{polylog}}
\newcommand{\Cop}{\mathrm{CO}}
\newcommand{\TC}{\mathrm{TC}}
\newcommand{\dmax}{d_{\max}}

\newcommand{\prob}[1]{\textnormal{\textsc{#1}}}

\newcommand{\HR}{\textnormal{\textsc{Heavy-Row}}}
\title{\bfseries Quantum Query Complexity of Tournament Solutions}
\author{Pragya Arora \qquad Udit Kumar\\
  \small Indian Institute of Technology Gandhinagar, India\\
  \small\texttt{\{pragya.arora, kumarudit\}@iitgn.ac.in}}
\date{ }

\begin{document}
\maketitle
\begin{abstract}
A tournament solution records a set of winners from a tournament, that is,
from the outcomes of all pairwise comparisons among $N$ alternatives. We study
the quantum query complexity of computing tournament solutions when the
outcome of each comparison has to be learned by a query. Classically, the
picture is uniform: the Copeland set, the uncovered set and the top cycle all
require $\Theta(N^2)$ queries, and so does finding a single vertex of maximum
out-degree. We show that the quantum query complexity of the top cycle can be computed
with $\tilde\Theta(N)$ queries, and a Condorcet winner can be found with $\tilde \Theta(\sqrt N)$ queries. Finding a
strong king requires $\Theta(N^{3/2})$ queries, the same as finding a vertex of
maximum out-degree. We also give a self-contained alternative proof of the known $\Theta(N^{3/2})$ bound for finding a vertex of maximum degree, established by Gilani, Mande, Patro, and Raja~\cite{GMPR26}. Our lower bounds reduce counting
problems to tournament problems using Ambainis's adversary method. For
the top cycle, we find a Hamiltonian path and locate the prefix forming
the top cycle by divide and conquer.
\end{abstract}

\section{Introduction}\label{sec:intro}
 
A \emph{tournament} is a directed graph in which every pair of distinct
vertices is joined by exactly one arc. Tournaments model round-robin
competitions, pairwise-majority relations in voting, and many other settings
in which alternatives are compared in pairs. A \emph{tournament solution} maps
every tournament to a nonempty set of ``winning'' vertices. Standard examples
include the Condorcet winner, the Copeland set ($\Cop(T)$), the uncovered set ($\textsc{UC}(T)$)
(equivalently, the set of kings), the Banks set and the top cycle ($\TC(T)$);
see~\cite{Las97,BBH16} for more details.
 
In many applications the tournament is only given implicitly: learning the
outcome of a single comparison may require an experiment, a clinical trial, or
a poll~\cite{MD24}. This motivates the \emph{query model}, in which an
algorithm learns the orientation of an edge by querying it and is charged one
unit per query. In this paper we study the \emph{quantum} query model, in
which queries may be made in superposition. Quantum query complexity is a
natural setting for such problems: it captures the power of quantum
algorithms such as Grover search and amplitude estimation, and it admits
strong lower-bound techniques such as the adversary method.

Maiti and Dey~\cite{MD24} showed that computing the Copeland set, the Slater
set, the Markov set, the bipartisan set, the uncovered set, the Banks set and
the top cycle requires $\Omega(N^2)$ queries, both deterministically and for
randomized algorithms. These lower bounds are proved by hiding a single
decisive edge among $\Theta(N^2)$ unqueried ones, which a quantum algorithm
can find with Grover search using $O(N)$ queries. For instance, the hard
instances for the Copeland set relate a regular tournament $T_{\mathrm{reg}}$
to the $\binom N2$ tournaments obtained by reversing one edge; in Ambainis's
adversary theorem (Theorem~\ref{thm:ambainis}) this gives $M=\binom N2$ and
$M'=L=L'=1$, and hence only an $\Omega(N)$ quantum lower bound. Consequently, the quantum complexity of each tournament solution has to be
revisited, and stronger quantum lower bounds require different constructions;
ours are based on counting.

Classically, the query complexity of tournament solutions is remarkably
uniform. Apart from the Condorcet winner, which can be found with $\Theta(N)$
queries~\cite{BRS97,Pro08}, all of the solutions listed above require
$\Theta(N^2)$ queries~\cite{MD24}, that is, reading a constant fraction of the
tournament, and even finding a single vertex of maximum out-degree requires
$\Theta(N^2)$ queries~\cite{GJR20}. Likewise, Ismaili Alaoui and Mande~\cite{AM25} recently showed that finding a
single strong king requires $\Theta(N^2)$ queries, answering a question of
Biswas et al.~\cite{BJRS}. In particular, the three classical solutions $\mathrm{CO}(T)\ \subseteq\ \mathrm{UC}(T)\ \subseteq\ \mathrm{TC}(T),$
namely, the Copeland set, the uncovered set and the top cycle, cannot be
distinguished by their classical query complexity. In the quantum setting, we show that the largest of the three, the top cycle,
is the easiest to compute, and the smallest, the Copeland set, is the hardest;
the uncovered set requires $\Omega(N^{3/2})$ and at most $O(N^2)$ queries.

\subsection{Results in this paper}\label{sec:results}
 The results established in this paper, together with previously known results, are summarized in Table~\ref{tab:results}. Throughout, \(N\) denotes the number of vertices, \(Q(\cdot)\) denotes bounded-error quantum query complexity, and \(\tilde{O}(\cdot)\) hides factors polylogarithmic in \(N\).
 
\begin{table}[H]
\centering
\small
\renewcommand{\arraystretch}{1.15}
\begin{tabular}{@{}llll@{}}
\toprule
Problem & Classical & Quantum & Result \\
\midrule
Condorcet winner
  & $\Theta(N)$~\cite{BRS97,Pro08}
  & $\tilde\Theta(\sqrt N)$
  & Thm.~\ref{thm:condorcet} \\
Top cycle
  & $\Theta(N^2)$~\cite{MD24}
  & $\tilde\Theta(N)$
  & Thm.~\ref{thm:topcycle-upper}, \ref{thm:topcycle-lb} \\
Max.\ (out-)degree vertex
  & $\Theta(N^2)$~\cite{GJR20}
  & $\Theta(N^{3/2})$~\cite{GMPR26}
  & Thm.~\ref{thm:maxdegree-general}$^{*}$, Thm.~\ref{thm:max-outdegreetournament}$^{*}$ \\
Strong king
  & $\Theta(N^{2})$ \cite{AM25}
  & $\Theta(N^{3/2})$
  & Thm.~\ref{thm:strongkings} \\
Uncovered set
  & $\Theta(N^2)$~\cite{MD24}
  & $\Omega(N^{3/2})$, $O(N^2)$
  & Thm.~\ref{thm:uncovered-lb} \\
Copeland set
  & $\Theta(N^2)$~\cite{MD24}
  & $\Theta(N^2)$
  & Thm.~\ref{thm:copeland} \\
\bottomrule
\end{tabular}
\caption{Query complexity of tournament solutions on $N$ vertices. Classical
bounds hold for deterministic and randomized algorithms. ${}^{*}$The quantum bound for the maximum degree is due
to~\cite{GMPR26}; we give a different proof.}
\label{tab:results}
\end{table}

Our most surprising result concerns the Copeland set, the set of all vertices
of maximum out-degree. We show that computing it requires $\Theta(N^2)$ quantum queries, matching the classical bound. This is in sharp contrast with a single Copeland winner, which can be found with
$\Theta(N^{3/2})$ quantum queries.
 
A Condorcet winner, if it exists, is the unique king of the tournament. Using
the $\tilde O(\sqrt N)$-query algorithm for finding a king of Mande,
Paraashar and Saurabh~\cite{MPS23}, followed by one Grover search, we find a
Condorcet winner or certify that none exists with $\tilde O(\sqrt N)$ queries,
and a reduction from $\textsc{OR}$ shows that $\Omega(\sqrt N)$ queries are
necessary. The same bounds hold for the Condorcet non-loser set.
 
For the top cycle we give an algorithm with $O(N\log^2 N)$ queries, which is
optimal up to logarithmic factors: the construction of Maiti and
Dey~\cite{MD24} embeds $\textsc{OR}$ on $\Theta(N^2)$ bits, which gives an
$\Omega(N)$ lower bound. Since a tournament is strongly connected if and only
if its top cycle contains every vertex, our algorithm also decides strong
connectivity of tournaments with $O(N\log^2N)$ queries. This is in contrast to
general directed graphs, for which strong connectivity requires
$\Theta(N^{3/2})$ queries in the adjacency-matrix model~\cite{DHHM06}.
 
The remaining problems depend on out-degrees, and their hardness comes from
counting rather than searching. The starting point is the problem of finding a
vertex of maximum degree in a graph, or of maximum out-degree in a tournament.
Gilani, Mande, Patro and Raja~\cite{GMPR26} recently showed that both
problems have quantum query complexity $\Theta(N^{3/2})$, answering a question
of~\cite{MPS23}. Their lower bound reduces from the composed function
$\textsc{OR}_m\circ\textsc{PrMaj}_m$ and relies on the composition theorem of
Lee et al.~\cite{LMRSS11}. In Section~\ref{sec:maxdeg} we give an alternative
proof of their result. We reduce from a promise problem that we call
\textsc{Heavy-Row}: given an $m\times m$ Boolean matrix in which every row has
weight $k$, except possibly one row of weight $k+1$, decide whether such a
heavy row exists. This problem is a restriction of
$\textsc{OR}_m\circ\textsc{PrMaj}$ (see Section~\ref{sec:prelim-heavyrow}),
and we prove its $\Omega(m^{3/2})$ lower bound directly with Ambainis's
adversary method, without the composition theorem. We include this proof
because it is self-contained and, more importantly, because our tournament
construction has additional structure on which the rest of the paper builds:
the same construction shows that finding a \emph{strong king}~\cite{HC03}
also requires $\Theta(N^{3/2})$ queries (Section~\ref{sec:strongking}), and a
variant with many heavy rows gives our $\Omega(N^2)$ lower bound for the
Copeland set (Section~\ref{sec:copeland}). The upper bound for strong kings is
immediate, since every vertex of maximum out-degree is a strong king. Classicaly, finding a strong king requires $\Theta(N^2)$ queries~\cite{AM25}, this is a polynomial quantum speedup, but a much
smaller one than for kings, which can be found with $\tilde O(\sqrt N)$
quantum queries~\cite{MPS23}. 
 
Finally, we show that computing the uncovered set requires $\Omega(N^{3/2})$
queries. This separates the uncovered set from two closely related problems:
a single king can be found with $\tilde O(\sqrt N)$ queries~\cite{MPS23}, and
the top cycle, which contains every king, can be computed with
$\tilde O(N)$ queries. Together with the Copeland set, these results show
that, quantumly, computing a whole tournament solution can be much harder than
finding one of its members: $\tilde\Theta(\sqrt N)$ versus $\Omega(N^{3/2})$
for kings, and $\Theta(N^{3/2})$ versus $\Theta(N^2)$ for Copeland winners.

\subsection{Related work}\label{sec:related}

\paragraph{Classical query complexity of tournament solutions:}
The query complexity of tournament solutions was studied systematically
by Maiti and Dey~\cite{MD24}. Earlier, Balasubramanian, Raman and 
Srinivasaragavan~\cite{BRS97}, and  Procaccia~\cite{Pro08} showed that deciding whether a tournament has a
Condorcet winner takes exactly $2N-\lfloor\log N\rfloor-2$ queries in the worst
case. Maiti and Dey proved
that every deterministic or randomized algorithm for the Copeland set, the
Slater set, the Markov set, the bipartisan set, the uncovered set, the Banks
set and the top cycle must make $\Omega(N^2)$ queries. Goyal, Jayapaul and Raman~\cite{GJR20} independently
proved the $\Omega(N^2)$ bound for the Copeland set; in fact, they showed that
even finding a single vertex of maximum out-degree requires $\Omega(N^2)$
queries, and more generally studied the query complexity of finding vertices
of given degrees. Arora, Dey and Misra~\cite{ADM26} showed that the classical $\Omega(N^2)$
barrier can be bypassed for tournaments close to transitive, given extra
structural queries. We instead keep arbitrary tournaments and ordinary edge
queries, and allow the queries to be quantum. We refer to~\cite{Las97,BBH16,Suk21} for more details
of tournament solutions.

\paragraph{Kings and strong kings:}
Kings were introduced by Landau~\cite{Lan53}, who showed that every vertex of
maximum out-degree is a king, so every tournament has a king. A tournament has
exactly one king if and only if this king is a Condorcet
winner~\cite{Mau80}, and there are tournaments in which every vertex is a
king~\cite{Rei82}. Shen, Sheng and Wu~\cite{SSW03} gave a deterministic
algorithm that finds a king with $O(N^{3/2})$ queries and showed that
$\Omega(N^{4/3})$ queries are necessary; closing this gap remains open. Mande,
Paraashar and Saurabh~\cite{MPS23} studied the randomized and quantum query
complexity of finding a king and showed in particular that
$\tilde O(\sqrt N)$ quantum queries suffice. Variants of this problem were
studied in~\cite{BJRS,LRT22,AGNS24}, and its communication complexity
in~\cite{MPSS24}. Strong kings were introduced by Ho and Chang~\cite{HC03},
who observed that every vertex of maximum out-degree is a strong king; see
also~\cite{CCCW08}. Biswas, Jayapaul, Raman and Satti~\cite{BJRS} asked
whether a strong king can be found with $o(N^2)$ queries, and Ismaili Alaoui
and Mande~\cite{AM25} recently answered this in the negative: the
deterministic and randomized query complexity of finding a strong king is
$\Theta(N^2)$. They also showed that deciding whether an arbitrary directed
graph has a king requires $\Theta(N^2)$ queries, in contrast to tournaments,
which always have one. Our $\Theta(N^{3/2})$ bound thus gives a polynomial quantum speedup for
finding a strong king, from $\Theta(N^2)$ to $\Theta(N^{3/2})$ queries; in
contrast, a king can be found with only $\tilde O(\sqrt N)$ quantum queries~\cite{MPS23}.

\paragraph{Quantum query complexity of graph problems:}
Dürr, Heiligman, Høyer and Mhalla~\cite{DHHM06} determined the quantum query
complexity of several graph problems, including connectivity and strong
connectivity, in the adjacency-matrix and adjacency-list models. Mande,
Paraashar and Saurabh~\cite{MPS23} asked about the quantum query complexity of
finding a vertex of maximum out-degree in a tournament. This question was answered by Gilani, Mande, Patro and Raja~\cite{GMPR26}, who proved an
$\Omega(N^{3/2})$ lower bound by combining the composition theorem of Lee et
al.~\cite{LMRSS11} with lower bounds for $\textsc{OR}$ and for a promise
version of majority; Section~\ref{sec:maxdeg} gives an alternative proof of
their result. Finally, our top-cycle algorithm starts by finding a Hamiltonian
path by insertion sort; the connection between sorting and Hamiltonian paths
in tournaments goes back to Bar-Noy and Naor~\cite{BNN90}.

\section{Preliminaries}\label{sec:prelims}
 
For a positive integer $n$, we write $[n]=\{1,\dots,n\}$. For a
string $x\in\{0,1\}^n$, we write $|x|$ for its Hamming weight.
 
\subsection{Tournaments and tournament solutions}\label{sec:prelim-tournaments}
 
Let $T=(V,E)$ be a tournament with $|V|=N$. We write $u\to v$ if $(u,v)\in E$
and say that $u$ \emph{beats} $v$. The \emph{out-degree} (or \emph{Copeland
score}) of $v$ is $d^+(v)=|\{w:v\to w\}|$, and $d_{\max}(T)=\max_v d^+(v)$. For distinct $u,v$ let
$p(u,v)=|\{w:u\to w\to v\}|$ be the number of directed paths of length two from
$u$ to $v$. A tournament is \emph{regular} if all its vertices have the same out-degree.
A standard example is the \emph{circulant tournament}: for odd $n$, it is the
tournament on the vertex set $\{v_0,\dots,v_{n-1}\}$ in which, for distinct
$i$ and $j$,
\[
  v_i\to v_j \iff (j-i)\bmod n\in\Bigl\{1,\dots,\tfrac{n-1}{2}\Bigr\}.
\]
 We study the following tournament solutions; see \cite{Las97,BBH16} for background and further examples.
 
\begin{definition}\label{def:solutions}
Let $T=(V,E)$ be a tournament.
\begin{itemize}
  \item A vertex $v$ is a \emph{Condorcet winner} if $d^+(v)=N-1$ and a
        \emph{Condorcet loser} if $d^+(v)=0$. The \emph{Condorcet non-loser
        set} is $\{v:d^+(v)\ge1\}$.
  \item The \emph{Copeland set} $\mathrm{CO}(T)$ is the set of vertices of
        maximum out-degree.
  \item A vertex $v$ is a \emph{king} if every other vertex is reachable from
        $v$ by a directed path of length at most two. The set of all kings is
        called the \emph{uncovered set} $\mathrm{UC}(T)$.
  \item A vertex $v$ is a \emph{strong king} \cite{HC03} if $p(v,u)>p(u,v)$
        for every $u$ with $u\to v$.
  \item A nonempty set $D\subseteq V$ is \emph{dominant} if every vertex of
        $D$ beats every vertex of $V\setminus D$. The \emph{top cycle}
        $\mathrm{TC}(T)$ is the unique smallest dominant set.
\end{itemize}
\end{definition}
 
 Every vertex $v$ of maximum out-degree is a king: if some vertex $u$ were not reachable from $v$ in at most two steps, then $u$ would beat $v$ and every out-neighbour of $v$, so $d(u)>d^+(v)$. Every king lies in the top cycle, since no vertex outside the top cycle can reach a vertex inside it.
Hence $\mathrm{CO}(T)\subseteq\mathrm{UC}(T)\subseteq\mathrm{TC}(T)$.

\begin{lemma}\label{lem:copeland-in-tc}
Every vertex of maximum out-degree in a tournament $T$ belongs to
$\mathrm{TC}(T)$.
\end{lemma}

We next give a degree-based characterization of strong kings. The
following identity relates the number of directed two-paths between
the endpoints of an arc to their out-degrees.

\begin{lemma}\label{lem:2paths}
If $u\to v$, then $p(v,u)-p(u,v)=d^+(v)-d^+(u)+1$.
\end{lemma}
 
\begin{proof}
Let $C=\{w:u\to w\text{ and }v\to w\}$ be the set of common out-neighbours of
$u$ and $v$. Since $u\to v$, the vertex $u$ is not an out-neighbour of $v$.
Every out-neighbour $w$ of $v$ therefore satisfies either $w\to u$, giving a
$2$-path $v\to w\to u$, or $u\to w$, in which case $w\in C$. Hence,
$d^+(v)=p(v,u)+|C|$.
 
The out-neighbours of $u$ are $v$ itself together with vertices $w\neq v$.
Each such $w$ satisfies either $w\to v$, giving a $2$-path $u\to w\to v$, or
$v\to w$, in which case $w\in C$. Hence, $d^+(u)=1+p(u,v)+|C|$.
 
Subtracting the two identities gives $p(v,u)-p(u,v)=d^+(v)-d^+(u)+1$.
\end{proof}

This identity gives the following characterization.

\begin{lemma}[Characterization]\label{lem:strongking-char}
A vertex $v$ is a strong king if and only if $d^+(u)\le d^+(v)$ for every $u$ with $u\to v$.
\end{lemma}
 
\begin{proof}
By Lemma~\ref{lem:2paths}, for $u\to v$ the condition $p(v,u)>p(u,v)$ is
equivalent to $d^+(v)-d^+(u)+1\ge1$, i.e.\ $d^+(u)\le d^+(v)$.
\end{proof}
 \begin{lemma}\label{lem:maxdeg-strongking}
Every vertex of maximum out-degree is a strong king. In particular, every
tournament has a strong king.
\end{lemma}
\begin{proof}
If $d^+(v)=d_{\max}(T)$, then $d^+(u)\le d^+(v)$ for every vertex $u$, so $v$ is a
strong king by Lemma~\ref{lem:strongking-char} below.
\end{proof}

Our quantum algorithm for the top cycle begins by finding a Hamiltonian
path. Hell and Rosenfeld~\cite{HR83} show that this can be
done using $O(N\log N)$ comparisons, which are edge queries in our
model.

\begin{lemma}[\cite{Sor88,HR83}]\label{lem:hampath}
A Hamiltonian path of an $N$-vertex tournament can be found by a
deterministic algorithm with $O(N\log N)$ queries.
\end{lemma}

 We also use a reachability description of the top cycle. The standard
characterization of $\mathrm{TC}(T)$ as the strongly connected
component that dominates all other strongly connected components
\cite[Sec.~3.3, pp.~71--72]{BBH16} gives the following lemma.
\begin{lemma}\label{lem:tc-reach}
Let $v\in\mathrm{TC}(T)$. Then:
\begin{enumerate}[label=(\alph*)]
  \item every in-neighbour of a vertex in $\mathrm{TC}(T)$ also belongs to
        $\mathrm{TC}(T)$;
  \item every vertex in $\mathrm{TC}(T)$ can reach $v$ by a directed path.
\end{enumerate}
Consequently, $\mathrm{TC}(T)=\{w\in V(T):w\text{ can reach }v\}$.
\end{lemma}
 
%\begin{proof}
%Let $C=\mathrm{TC}(T)$. Since $C$ dominates every vertex outside $C$, no vertex outside $C$ can have an edge to a vertex in $C$. This proves (a). The set $C$ is strongly connected, so every vertex in $C$ can reach $v$.Conversely, no vertex outside $C$ can reach $v$, since any such path would have to contain an edge from outside $C$ into $C$. No such edge exists. Thus, $\mathrm{TC}(T)=\{w\in V(T):w\text{ can reach }v\}$.
%\end{proof}
\paragraph{Quantum query model.}
We use the standard quantum query model. A tournament $T$ on the vertex set
$[N]$ is given by a string in $\{0,1\}^{\binom N2}$, with one bit for each pair
of distinct vertices that specifies the orientation of the edge between them.
It is accessed through the oracle
$O_T\colon|u,v,b\rangle\mapsto|u,v,b\oplus[u\to v]\rangle$, where
$[u\to v]\in\{0,1\}$ indicates whether $u$ beats $v$; a graph is accessed
analogously through its adjacency matrix, and a string $x\in\{0,1\}^n$
through $O_x\colon|i,b\rangle\mapsto|i,b\oplus x_i\rangle$. The cost of an
algorithm is the number of oracle applications, and the
\emph{bounded-error quantum query complexity} $Q(P)$ of a problem $P$ is the
smallest cost of an algorithm that, on every input, outputs a correct answer
with probability at least $2/3$. In our lower bounds we build a tournament
$T_x$ from an input $x$ such that every edge of $T_x$ is either fixed or
determined by a single bit of $x$. A query to $T_x$ can then be simulated with
at most one query to $x$, so an algorithm on $T_x$ that makes $q$ queries
gives an algorithm on $x$ that makes at most $q$ queries.

We recall several standard quantum query bounds that we use throughout the paper.
\begin{theorem}[\cite{BBCMW01}]\label{thm:or}
For every positive integer $n$, the bounded-error quantum query complexity of
$\textup{\textsc{OR}}_n$ is $\Theta(\sqrt n)$.
\end{theorem}
 
\begin{theorem}[\cite{MPS23}]\label{thm:king}
A king of an $N$-vertex tournament can be found with bounded error using
$O(\sqrt N\,\mathrm{polylog}\,N)$ quantum queries.
\end{theorem}
 
%\begin{theorem}[\cite{BBHT98}]\label{thm:counting}
%For every $X\in\{0,1\}^N$, there is a quantum algorithm that returns $t=|X|$
%with probability at least $3/4$ and uses expected time
%$\Theta\bigl(\sqrt{(t+1)(N-t+1)}\bigr)$.
%\end{theorem}

Finally, we use the following search guarantee when the number of marked inputs is unknown.
 
\begin{theorem}[\cite{BBHT98}, Theorem~3]\label{thm:bbht}
Let $f\colon[n]\to\{0,1\}$ and $t=|f^{-1}(1)|$, where $t$ is unknown. There is
a quantum algorithm that, if $t\ge1$, outputs a marked element using an
expected $O(\sqrt{n/t})$ queries to $f$. Moreover, for every $\delta>0$,
stopping it after $O(\sqrt n\log(1/\delta))$ queries gives an algorithm that
reports ``no marked element'' when $t=0$, and errs with probability at most
$\delta$ when $t\ge1$.
\end{theorem}

\subsection{The Heavy-Row problem and the adversary framework}
\label{sec:prelim-heavyrow}
 
To prove lower bounds for several problems studied in this paper, we use the following promise problem, which we call the \textsc{Heavy-Row} problem. It captures the single-row, unit-gap signal used in our reductions: a row of weight $k+1$ will correspond to a vertex whose (out-)degree is one larger than that of a vertex corresponding to a row of weight $k$.
 
\begin{definition}[\textsc{Heavy-Row}]\label{def:heavyrow}
Fix an odd integer $m\ge3$ and let $k=\frac{m-1}{2}$. An input is an
$m\times m$ Boolean matrix $x=(x_{ij})\in\{0,1\}^{m\times m}$, so there are
$n=m^2$ Boolean variables. For each row $i$, let $x_i=(x_{i1},\dots,x_{im})$
and let $|x_i|$ denote its Hamming weight. Let
\begin{align*}
  R_0&=\{x:|x_i|=k\text{ for every }i\},\\
  R_1&=\{x:\text{there is a unique }i^\star\text{ with }|x_{i^\star}|=k+1,
        \text{ and }|x_i|=k\text{ for every }i\neq i^\star\}.
\end{align*}
Inputs in $R_0$ are no-instances and inputs in $R_1$ are yes-instances; in a
yes-instance, $i^\star$ is called the \emph{heavy row}. Given
$x\in R_0\cup R_1$, the task is to decide which case holds.
\end{definition}
 
\begin{example}\label{ex:heavyrow}
For $m=3$ we have $k=1$. The no-instance
\[
  x=\begin{pmatrix}1&0&0\\0&1&0\\0&0&1\end{pmatrix}
\]
has $6$ zeros, and changing any one of them into a $1$ gives a yes-instance,
so $M=6$. Flipping the top-middle entry gives the yes-instance
\[
  y=\begin{pmatrix}1&1&0\\0&1&0\\0&0&1\end{pmatrix},
\]
from which the only way back into $R_0$ is to clear one of the two $1$s of the
top row, so $M'=2$. The bound of Theorem~\ref{thm:ambainis} is
$\sqrt{6\cdot2}=\sqrt{12}=2\sqrt3\approx3.46$.
\end{example}
 
\paragraph{Relation to prior work.}
\textsc{Heavy-Row} is a restricted version of the composed function
$\textsc{OR}_m\circ\textsc{PrMaj}$ used by Gilani, Mande, Patro and
Raja~\cite{GMPR26} in their lower bounds for maximum degree. There,
$\textsc{PrMaj}$ (promise majority) distinguishes strings of even length
$\ell$ and Hamming weight $\ell/2$ from strings of weight $\ell/2+1$, and one
asks whether at least one of $m$ such strings is heavy. \textsc{Heavy-Row}
differs only in that the row length is odd, so that the two weights are $k$
and $k+1$, and in the additional promise that at most one row is heavy.
Gilani et al.\ obtain an $\Omega(m^{3/2})$ lower bound for
$\textsc{OR}_m\circ\textsc{PrMaj}_m$ by combining the composition theorem of
Lee et al.~\cite{LMRSS11} with the bounds
$Q(\textsc{OR}_m)=\Omega(\sqrt m)$~\cite{BBBV97} and
$Q(\textsc{PrMaj}_m)=\Omega(m)$~\cite{Amb02}. The composition theorem as
stated there requires the outer function to be total, so it does not
immediately give a bound under the promise that at most one row is heavy. In
Section~\ref{sec:heavyrow-lb} we instead give a short, self-contained proof
for \textsc{Heavy-Row} that uses only Ambainis's adversary theorem. Since
appending a fixed $1$ to every row maps an instance of \textsc{Heavy-Row} to
an instance of $\textsc{OR}_m\circ\textsc{PrMaj}_{m+1}$ with the same answer,
our bound also yields an $\Omega(m^{3/2})$ lower bound for
$\textsc{OR}_m\circ\textsc{PrMaj}_{m+1}$.
 
We will show that the quantum query complexity of \textsc{Heavy-Row} is
$\Omega(m^{3/2})$ by applying Ambainis's general lower-bound
theorem~\cite{Amb02}.
 
\begin{theorem}[Ambainis~\cite{Amb02}]\label{thm:ambainis}
Let $f(x_1,\dots,x_n)$ be a function of $n$ Boolean variables, and let $X,Y$
be sets of inputs such that $f(x)\neq f(y)$ whenever $x\in X$ and $y\in Y$. Let
$R\subseteq X\times Y$ be such that:
\begin{enumerate}[label=(\arabic*)]
  \item for every $x\in X$ there are at least $M$ different $y\in Y$ with
        $(x,y)\in R$;
  \item for every $y\in Y$ there are at least $M'$ different $x\in X$ with
        $(x,y)\in R$;
  \item for every $x\in X$ and $i\in[n]$ there are at most $L$ different
        $y\in Y$ with $(x,y)\in R$ and $x_i\neq y_i$;
  \item for every $y\in Y$ and $i\in[n]$ there are at most $L'$ different
        $x\in X$ with $(x,y)\in R$ and $x_i\neq y_i$.
\end{enumerate}
Then every bounded-error quantum algorithm computing $f$ uses
$\Omega\bigl(\sqrt{MM'/(LL')}\bigr)$ queries.
\end{theorem}

\section{A Query Lower Bound for the Heavy-Row Problem} \label{sec:heavyrow-lb}

In this section, we prove a lower bound on the quantum query complexity of \HR~by applying Ambainis's theorem (Thm.~\ref{thm:ambainis}).

\begin{theorem}\label{thm:heavyrow}
$Q(\HR) = \Omega(m^{3/2})$.
\end{theorem}

\begin{proof}
Apply Theorem~\ref{thm:ambainis} with $X$ the set of no-instances, $Y$ the set of yes-instances, and $
  R = \{(x,y)\in X\times Y : y \text{ is obtained from } x \text{ by changing a single } 0 \text{ into a } 1\}.$

\emph{Parameter $M$.} Let $x\in X$. Each row of $x$ has $m-k$ zeros, so $x$ has $m(m-k)$ zeros in total. Changing any one of them into a $1$ raises that row to weight $k+1$ and leaves the other rows at weight $k$, which gives a yes-instance. Hence $M=m(m-k)$.

\emph{Parameter $M'$.} Let $y\in Y$ with heavy row $i^\star$. A related $x\in X$ is obtained by changing one $1$ of $y$ into a $0$. If the changed $1$ lies in row $i^\star$, every row of the result has weight $k$, so the result lies in $X$; there are $k+1$ such choices. If it lies in another row, that row drops to weight $k-1$ and the result is not a valid input. Hence $M'=k+1$.

\emph{Parameters $L$ and $L'$.} Related inputs differ in exactly one position. So for fixed $x$ and a fixed position $(i,j)$, the only related $y$ that can differ from $x$ at $(i,j)$ is the one obtained by flipping $(i,j)$ itself. Hence $L=1$, and by the same argument $L'=1$. Since $m-k=k+1=\frac{m+1}{2}$, Theorem~\ref{thm:ambainis} gives $Q(\HR) = \Omega(m^{3/2}).$
\end{proof}

\section{Maximum degree} \label{sec:maxdeg}
In this section, we study two maximum-degree problems. In
\textsc{MaxDeg-Vertex}, given an $N$-vertex undirected graph, the task is to output a vertex of maximum degree. Its quantum query complexity is
$\Theta(N^{3/2})$~\cite{GMPR26}; we give a self-contained proof, which also extends to directed graphs. We also study \textsc{MaxOutDeg-Tournament}: given an $N$-vertex tournament, the task is to output a vertex of maximum out-degree. We prove the same $\Theta(N^{3/2})$ bound for this problem.

\subsection{Maximum degree in general graphs}

 To obtain the quantum upper bound, computing the degree of each vertex classically requires $N-1$ queries, and finding the maximum among the $N$ degrees takes $O(\sqrt{N})$ queries \cite{DH96}. Thus, the total query complexity is $O(N^{3/2})$. For the lower bound we prove the following proposition.
\begin{proposition}\label{prop:general-graphs}
Finding a vertex of maximum degree in an $N$-vertex undirected graph requires $\Omega(N^{3/2})$.
\end{proposition}

\begin{proof}

We will reduce this problem from the \prob{Heavy-Row} problem. Let $x=(x_{ij})_{m \times m}$ be an instance of the \prob{Heavy-Row} problem.

\noindent \emph{Construction of the undirected graph:} Construct an undirected graph $G_x$ with $N=3m$ vertices. Let $U$, $W$, and $P$ be a partition of its vertex set where $U=\{u_0,\dots,u_{m-1}\}$, $W=\{w_0,\dots,w_{m-1}\}$, and $P=\{p_0,\dots,p_{m-1}\}$.

The edge set of $G_x$ is defined as follows:
\begin{itemize}
    \item $\{u_i,w_j\}$ if and only if  $x_{ij}=1$
    \item $\{u_i,p_l\}$ for all $i,l$
\end{itemize}

Each query to $G_x$ costs at most one query to $x$. If every row of $x$ has weight $k$ or $k+1$, then
 $d(u_i)=m+|x_i|\in\{m+k,\,m+k+1\}$, $d(w_j)=|\{i:x_{ij}=1\}|\le m$, and $d(p_l)=m$.

Since $k\ge1$, the vertices of maximum degree are all of $u_0,\dots,u_{m-1}$ if $x\in\mathcal R_0$, and only $u_{i^\star}$ if $x\in\mathcal R_1$ with heavy row $i^\star$.

Therefore, finding a vertex of maximum degree in the constructed graph distinguishes between the two cases of the Heavy-Row problem. Since the Heavy-Row problem requires $\Omega(m^{3/2})$ queries, finding a maximum-degree vertex also requires $\Omega(m^{3/2})$ queries. As $N=3m$, this gives a lower bound of $\Omega(N^{3/2})$. For values of $N$ not divisible by three, we can add isolated vertices without affecting the argument.
\end{proof}

Combining this lower bound with the $O(N^{3/2})$ upper bound gives a query complexity of $\Theta(N^{3/2})$ for finding a maximum-degree vertex in both undirected and directed graphs.
\begin{theorem}\label{thm:maxdegree-general}
    $Q($\prob{MaxDeg-Vertex}$)=\Theta(N^{3/2})$ in both undirected and directed graphs.
\end{theorem}

\subsection{Maximum Out Degree in Tournament}

The quantum query complexity of finding the maximum out-degree vertex in a tournament is the same as in general graphs, that is, $\Theta(N^{3/2})$. The upper bound follows from the same argument as in the general case.
\begin{proposition}\label{prop:maxdeg-lower}
Finding a vertex of maximum out-degree in an $N$-vertex tournament requires $\Omega(N^{3/2})$ quantum queries.
\end{proposition}
\begin{proof}
We will reduce this problem from the \prob{Heavy-Row} problem. Let $x=(x_{ij})_{m \times m}$ be an instance of the \prob{Heavy-Row} problem.

\noindent \emph{Construction of the Tournament ($T_x$):}
Let $N=10m+1$, where $m\ge3$ is odd. Let $A, B,$ and $C$ be the partition of its vertex set with $|A|=3m$, $|B|=m$, and $|C|=6m+1$.  Since $m$ is odd, we can construct a circulant tournament within each set.

Furthermore, partition $A$ into the triples $A_i = \{a_i, a_{i+m}, a_{i+2m}\},$ for $0\leq i <m$, where $A_i$ consists of the vertices $a_s$ such that $s \equiv i (\bmod m)$.

Add the arcs between the sets as shown in Figure \ref{fig:groups}:

\begin{figure}[ht]
\centering
\begin{tikzpicture}[
  box/.style={draw, rounded corners, minimum width=3.3cm, minimum height=1.35cm, align=center},
  >=Stealth]
  \node[box] (A) at (0,0) {Set $A$\\[2pt] $3m$ vertices};
  \node[box] (C) at (4.6,0) {Set $C$\\[2pt] $6m+1$ vertices};
  \node[box] (B) at (9.2,0) {Set $B$\\[2pt] $m$ vertices};
  \draw[->, thick] (A) -- node[above, font=\small]{all} (C);
  \draw[->, thick] (C) -- node[above, font=\small]{all} (B);
  \draw[dashed, thick] (A.north) to[bend left=22]
    node[above, font=\small]{for $a\in A_i$:\ $a\to b_j \iff x_{ij}=1$} (B.north);
\end{tikzpicture}
\caption{The three vertex sets of \(T_x\). Solid arrows indicate that every vertex in one set dominates every vertex in the other set in the indicated direction. The dashed arrow represents the only input-dependent edges.}
\label{fig:groups}
\end{figure}
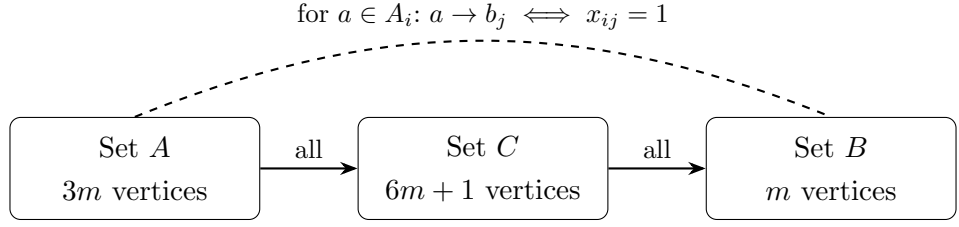

\noindent \textbf{Out-degrees:} For $a\in A_i$:
\[
  d^+(a) = \underbrace{\tfrac{3m-1}{2}}_{\text{inside } A} + \underbrace{|x_i|}_{\text{vs.\ } B} + \underbrace{(6m+1)}_{\text{vs.\ } C}
  = \begin{cases}
      8m, & |x_i| = k,\\
      8m+1, & |x_i| = k+1,
    \end{cases}
\]

For $b_j\in B$: $b_j$ beats exactly the three vertices of each triple $A_i$ with $x_{ij}=0$, so
\[
  d^+(b_j) = \tfrac{m-1}{2} + 3\,\bigl|\{i : x_{ij}=0\}\bigr| \le \tfrac{m-1}{2} + 3m = \tfrac{7m-1}{2} < 4m.
\]

For $c\in C$: $d^+(c) = 3m + m = 4m$.

Thus, if $x\in\mathcal R_0$, then $\dmax^+(T_x)=8m$, attained by every vertex of $A$.
If $x\in\mathcal R_1$ with heavy row $i^\star$, then $\dmax^+(T_x)=8m+1$, and the vertices of maximum out-degree are exactly the three vertices of $A_{i^\star}$.

Therefore, finding a vertex of maximum out-degree in the constructed tournament allows us to solve the Heavy-Row problem, as in the general-graph case. Since the Heavy-Row problem requires $\Omega(m^{3/2})$ queries, finding a maximum out-degree vertex also requires $\Omega(m^{3/2})$ queries. As $N=10m+1$, this gives a lower bound of $\Omega(N^{3/2})$.

For values of $N$ that are not of the form $10m+1$ with $m$ odd. Let $N\ge31$, let $m$ be the largest odd integer with $10m+1\le N$, and let $n_0=N-(10m+1)$. Then $m\ge3$, and since $10(m+2)+1>N$, we have $n_0\le19$, so $m=\Theta(N)$. Let $T'_x$ be the tournament obtained from $T_x$ by adding $n_0$ new vertices, ordered transitively among themselves, each of which loses to every vertex of $T_x$.
Every vertex of $T_x$ gains exactly $n_0$ out-neighbours, so the order of the out-degrees within $T_x$ is unchanged, and every vertex of $T_x$ has out-degree at least $n_0$ in $T'_x$. Every new vertex beats only new vertices, so its out-degree is at most $n_0-1$. Hence $T'_x$ and $T_x$ have the same vertices of maximum out-degree. The new edges do not depend on $x$, so the argument above applies to $T'_x$ without change, and finding a vertex of maximum out-degree in an $N$-vertex tournament requires $\Omega(m^{3/2})=\Omega(N^{3/2})$ queries for every $N\ge31$. Since the bound is asymptotic, this completes the proof.
\end{proof}

\noindent Combining the upper and lower bounds, we obtain the following theorem.

\begin{theorem}\label{thm:max-outdegreetournament}
    $Q($\prob{MaxOutDeg-Tournament}$)=\Theta(N^{3/2})$.
\end{theorem}

\section{Strong Kings}\label{sec:strongking}
In this section, we prove that the quantum query complexity of finding a strong king in a tournament is $\Theta(N^{3/2})$. By Lemma \ref{lem:maxdeg-strongking} and Theorem \ref{thm:max-outdegreetournament}, the upper bound follows directly. For the lower bound, the construction which we use for maximum out-degree is the same. To complete the argument for strong kings, we need the following observation and lemma. 

\noindent Since $A$ is a circulant tournament it is easy to observe the following.
\begin{description}\label{Obs:triple}
    \item[Observation:] For every vertex $a_s\in A$ and every $i\in\{0,\dots,m-1\}$, some vertex of $A_i$ beats $a_s$.
\end{description}

Using the above Observation, we obtain the following lemma.
\begin{lemma}\label{lem:strongkings}
\begin{enumerate}[label=(\alph*)]
  \item If $x$ is a no-instance, the strong kings of $T_x$ are exactly the vertices of $A$.
  \item If $x$ is a yes-instance with heavy row $i^\star$, the strong kings of $T_x$ are exactly the three vertices of $A_{i^\star}$.
\end{enumerate}
\end{lemma}

\begin{proof}
We use Lemma~\ref{lem:strongking-char} throughout.

\emph{No vertex of $B$ is a strong king:} each $b\in B$ is beaten by every $c\in C$, and $d^+(c)=4m>d^+(b)$.

\emph{No vertex of $C$ is a strong king:} each $c\in C$ is beaten by every $a\in A$, and $d^+(a)\ge 8m > 4m = d^+(c)$.

In a no-instance, every vertex of $A$ has out-degree $8m$, the maximum in $T_x$. By Lemma~\ref{lem:maxdeg-strongking}, all of them are strong kings.

In a yes-instance, the vertices of $A_{i^\star}$ have the unique maximum out-degree $8m+1$, so they are strong kings by Lemma~\ref{lem:maxdeg-strongking}. Every other $a_s\in A\setminus A_{i^\star}$ has out-degree $8m$. By Observation~\ref{Obs:triple} some $a_r\in A_{i^\star}$ beats $a_s$, and $d^{+}(a_r)=8m+1>d^+(a_s)$, so $a_s$ is not a strong king.
\end{proof}

\begin{theorem}\label{thm:strongkings}
$Q($\prob{StrongKings}$)$ = $\Theta(N^{3/2})$.
\end{theorem}

\begin{proof}
Let $\mathcal A$ be a quantum algorithm that outputs a strong king of any $N$-vertex tournament with probability at least $2/3$ as follows.
\begin{enumerate}[label=\arabic*.]
  \item Run $\mathcal A$ on $T_x$ (with $N=10m+1$), simulating each tournament query with $O(1)$ queries to $x$. Let $v$ be its output.
  \item If $v\in A_i$ for some $i$, read the $m$ bits of row $i$ and answer \emph{yes} if and only if $|x_i|=k+1$.
  \item If $v\notin A$, answer \emph{no}.
\end{enumerate}

\emph{Correctness.} On a no-instance every row has weight $k$, so the algorithm always answers \emph{no}, whatever $\mathcal A$ outputs. On a yes-instance, with probability at least $2/3$ the output $v$ is a strong king, hence $v\in A_{i^\star}$ by Lemma~\ref{lem:strongkings}(b), and the row check answers \emph{yes}. So the error is at most $1/3$.

\emph{Cost.} Step~1 makes at most $q$ queries to $x$, and Step~2 makes $m$
queries. Hence the algorithm solves \prob{Heavy-Row} with at most $q+m$
queries. Since \prob{Heavy-Row} requires $\Omega(m^{3/2})$ queries
(Theorem~\ref{thm:heavyrow}), we get $q+m=\Omega(m^{3/2})$, that is,
$q=\Omega(m^{3/2})=\Omega(N^{3/2})$.

For general $N\ge31$, we use the tournament $T'_x$ from the
proof of Proposition~\ref{prop:maxdeg-lower}, obtained from $T_x$ by adding $n_0\le19$ new vertices, ordered transitively among themselves, each of which loses to every vertex of $T_x$. It remains to check that $T'_x$ and $T_x$ have the same strong kings, for which we use Lemma~\ref{lem:strongking-char}. Let $v$ be a vertex of $T_x$. Every new vertex loses to $v$, so $v$ has the same in-neighbours in $T'_x$ as in $T_x$, and every out-degree in $T_x$ increases by
exactly $n_0$; hence $v$ is a strong king of $T'_x$ if and only if it is a strong king of $T_x$. A new vertex $w$ is beaten by every vertex of $T_x$, and these vertices have out-degree at least $n_0>d^+(w)$ in $T'_x$, so $w$ is not a strong king. Therefore the algorithm above works with $T'_x$ in place of $T_x$, and $q=\Omega(m^{3/2})=\Omega(N^{3/2})$ for every $N\ge31$, since
$m=\Theta(N)$.
\end{proof}

\section{Quantum query complexity of Tournament Solutions}

\subsection{Condorcet Winner}

In this problem, we must output the Condorcet winner of an $N$-vertex tournament if it exists; otherwise, we report that no Condorcet winner exists.

\begin{lemma}\label{lem:unique-king}
If $w$ is a Condorcet winner of $T$, then $w$ is the unique king of $T$.
\end{lemma}

\begin{proof}
The vertex $w$ reaches every vertex in one step, so it is a king. Since $w$ has in-degree $0$, no other vertex can reach $w$ by a path of any length, so no other vertex is a king.
\end{proof}

\noindent Since a Condorcet winner is the unique king, the following theorem gives quantum query bounds for finding a Condorcet winner.

\begin{theorem}\label{thm:condorcet}
The bounded-error quantum query complexity of finding a Condorcet winner in an $N$-vertex tournament is $\Omega(\sqrt{N})$ and $O(\sqrt{N}\polylog N)$.
\end{theorem}

\begin{proof}
The algorithm first finds a king $v$ using Theorem \ref{thm:king}, then uses Grover search to find a vertex $u$ with $u\to v$ in $O(\sqrt N)$ queries. If none exists, it outputs $v$; otherwise, it reports that no Condorcet winner exists.

For the lower bound, we reduce from $\mathrm{OR}_m$, whose bounded-error quantum query complexity
is $\Omega(\sqrt{m})$~\cite{BBCMW01}. Let $m=N-1\ge 3$ and let
$y\in\{0,1\}^m$ be an input to $\mathrm{OR}_m$. We construct a tournament
$T_y$ on the vertex set $V=\{w\}\cup U$, where $U=\{u_1,\dots,u_m\}$.

\emph{Edges inside $U$.} We fix a tournament on $U$ in which every vertex has
in-degree at least one: we let $u_i\to u_{i+1 \bmod m}$ for every $i$, and
orient the remaining pairs of $U$ arbitrarily but fixed, independently of $y$.

\emph{Edges at $w$.} For every $i\in[m]$, we let $u_i\to w$ if $y_i=1$, and
$w\to u_i$ if $y_i=0$.

\emph{Correctness.} No vertex $u_i$ is a Condorcet winner of $T_y$, because it
is beaten by $u_{i-1 \bmod m}$. The vertex $w$ is a Condorcet winner if and
only if it beats every $u_i$, that is, if and only if $y_i=0$ for every $i$. Thus, $T_y$ has a Condorcet winner if and only if $y=0^m$, which is if and only if $\mathrm{OR}_m(y)=0$.

Thus, running any $T$-query bounded-error quantum algorithm for deciding whether an $N$-vertex tournament has a Condorcet winner on $T_y$ and negating its output computes $\mathrm{OR}_m(y)$ using $T$ queries. Hence, $T=\Omega(\sqrt{m})=\Omega(\sqrt{N})$.
\end{proof}

\begin{corollary}\label{cor:nonloser}
The Condorcet non-loser set can be computed with $O(\sqrt N\polylog N)$ quantum queries, and $\Omega(\sqrt N)$ queries are necessary.
\end{corollary}
\subsection{Copeland Set}\label{sec:copeland}

The Copeland set of a tournament is the set of all its vertices of maximum
out-degree. Finding one such vertex was treated in Section~\ref{sec:maxdeg};
here we compute the whole set. Querying all $\binom N2$ edges computes it, so
the problem is only interesting for lower bounds, and we show that this
trivial algorithm is optimal up to a constant factor. For the lower bound, we
generalize \textsc{Heavy-Row} from one heavy row to many.

\begin{definition}[\textsc{Heavy-Rows$_{k}$}]\label{def:heavyrows}
Fix an odd integer $m\ge3$ and let $k=(m-1)/2$, as in
Definition~\ref{def:heavyrow}. For $r\in\{0,\dots,m\}$, let $R_r$ be the set
of $m\times m$ Boolean matrices with exactly $r$ rows of weight $k+1$ and the
remaining $m-r$ rows of weight $k$; we call the rows of weight $k+1$
\emph{heavy}. In \textsc{Heavy-Rows}, we are given $x\in R_k\cup R_{k+1}$ and
must decide whether $x\in R_k$ or $x\in R_{k+1}$.
\end{definition}

\begin{lemma}\label{lem:heavyrows-lb}
$Q(\textsc{Heavy-Rows})=\Omega(m^2)$.
\end{lemma}

\begin{proof}
We apply Theorem~\ref{thm:ambainis} with $X=R_k$ and $Y=R_{k+1}$, relating
$x\in X$ and $y\in Y$ if $y$ is obtained from $x$ by changing a single $0$
into a $1$. Note that $m-k=k+1=(m+1)/2$.

\emph{Parameter $M$.} Let $x\in X$. It has $m-k=k+1$ non-heavy rows, and each
of them has $m-k=k+1$ zeros. Changing one of these zeros into a $1$ makes that
row heavy and yields an element of $Y$, while changing a $0$ in a heavy row
creates a row of weight $k+2$. Hence $M=(k+1)^2$.

\emph{Parameter $M'$.} Let $y\in Y$. It has $k+1$ heavy rows, and each of them
has $k+1$ ones. Changing one of these ones into a $0$ yields an element of
$X$, while changing a $1$ in a non-heavy row creates a row of weight $k-1$.
Hence $M'=(k+1)^2$.

\emph{Parameters $L$ and $L'$.} Related inputs differ in exactly one position,
so $L=L'=1$.

Theorem~\ref{thm:ambainis} gives the bound
$\sqrt{MM'/(LL')}=(k+1)^2=\bigl((m+1)/2\bigr)^2=\Omega(m^2)$.
\end{proof}

We use the tournament $T_x$ on $N=10m+1$ vertices from the proof of
Proposition~\ref{prop:maxdeg-lower}.

\begin{lemma}\label{lem:copeland-tx}
Let $x\in R_r$ with $r\ge1$. Then $\mathrm{CO}(T_x)$ is the union of the
triples $A_i$ of the heavy rows $i$ of $x$, so $|\mathrm{CO}(T_x)|=3r$.
\end{lemma}

\begin{proof}
By the proof of Proposition~\ref{prop:maxdeg-lower}, every vertex of a triple
$A_i$ has out-degree $8m+1$ if row $i$ is heavy and $8m$ otherwise, and every
vertex of $B\cup C$ has out-degree less than $8m$. Since $r\ge1$, the maximum
out-degree is $8m+1$, and it is attained exactly by the vertices of the
triples of the heavy rows.
\end{proof}

\begin{theorem}\label{thm:copeland}
The bounded-error quantum query complexity of computing the Copeland set of
an $N$-vertex tournament is $\Theta(N^2)$.
\end{theorem}
 
\begin{proof}
The upper bound $O(N^2)$ holds since an algorithm can query all $\binom N2$
edges.
 
For the lower bound, first let $N=10m+1$ for an odd integer $m\ge3$, and let
$k=(m-1)/2\ge1$. Let $x\in R_k\cup R_{k+1}$. By Lemma~\ref{lem:copeland-tx},
\[
  |\mathrm{CO}(T_x)|=
  \begin{cases}
    3k,   & x\in R_k,\\
    3k+3, & x\in R_{k+1}.
  \end{cases}
\]
Let $\mathcal A$ be a bounded-error quantum algorithm that computes the
Copeland set using $q$ queries. We run $\mathcal A$ on $T_x$, simulating each
query to $T_x$ with at most one query to $x$, and answer that $x\in R_{k+1}$ if
and only if the returned set has $3k+3$ vertices. This solves
\textsc{Heavy-Rows} with at most $q$ queries, so $q=\Omega(m^2)=\Omega(N^2)$ by
Lemma~\ref{lem:heavyrows-lb}.
 
Now let $N\ge31$ be arbitrary, let $m$ be the largest odd integer with
$10m+1\le N$, so that $m\ge3$ and $m=\Theta(N)$, and let $n_0=N-10m-1$. We
add $n_0$ new vertices to $T_x$, ordered transitively among themselves, each
of which loses to every vertex of $T_x$. Every vertex of $T_x$ gains exactly
$n_0$ out-neighbours, so the order of the out-degrees within $T_x$ is
unchanged, and every vertex of $T_x$ now has out-degree at least $n_0$. Each
new vertex beats only new vertices, so its out-degree is at most $n_0-1$.
Hence, the Copeland set does not change. The new edges do not depend on $x$, so
each query is still simulated with at most one query to $x$, and the bound
$\Omega(m^2)=\Omega(N^2)$ holds for every $N\ge31$. For $N<31$ the statement
is trivial.
\end{proof}

\section{Top Cycle}\label{sec:topcycle}
In this section, we establish quantum upper and lower bounds for finding the top cycle of a tournament. Classically, computing the top cycle requires $\Theta(N^2)$ queries~\cite{MD24}. We show that its quantum query complexity is $\tilde{\Theta}(N)$. Our algorithm first finds a Hamiltonian path. The top cycle forms a prefix of this path, and finding the end of this prefix reduces to an interval-covering problem, which we solve using divide-and-conquer.

Fix a Hamiltonian path $v_1\to v_2\to\dots\to v_N$. For $1\le p\le N-1$,
\emph{cut} $p$ separates $\{v_1,\dots,v_p\}$ from $\{v_{p+1},\dots,v_N\}$.
A \emph{back edge} is a pair $(i,j)$ with $i<j$ and $v_j\to v_i$. Since
$v_i\to v_{i+1}$, every back edge satisfies $j\ge i+2$. The back edge
$(i,j)$ \emph{covers} cut $p$ if $i\le p<j$, that is, it covers the cuts
$i,i+1,\dots,j-1$.

\begin{lemma}\label{lem:prefix}
Let $t$ be the smallest cut not covered by any back edge, then $\TC(T)=\{v_1,\ldots,v_t\}$. Further, if every cut is covered by a back edge, then $\TC(T)=V(T)$.
\end{lemma}

\begin{proof}
First, every dominant set $D$ is a prefix of the path: if $v_j\in D$ and
$v_{j-1}\notin D$ for some $j\ge2$, then $v_j\to v_{j-1}$, contradicting
$v_{j-1}\to v_j$. Second, the prefix $\{v_1,\dots,v_p\}$ with $p\le N-1$ is
dominant if and only if no vertex after cut $p$ beats a vertex before it,
that is, if and only if no back edge covers cut $p$. Since $\TC(T)$ is the
smallest dominant set, the claim follows.
\end{proof}

By Lemmas~\ref{lem:hampath} and~\ref{lem:prefix}, it remains to find the smallest uncovered cut.
We query $b(i,j)=[v_j\to v_i]$, where $i+2\le j$; each query costs one.

Represent each back edge $(i,j)$ as the grid point in column $i$ and
row $j$. For a set $C$ of columns, define
$\operatorname{reach}(C)=\max\{j:b(i,j)=1\text{ for some }i\in C\}$,
with $\max\emptyset=-\infty$.\label{eq:reach} Then, for every $p\ge l$,
cut $p$ is covered by a back edge in a column less than $l$ if and only if
$\operatorname{reach}([1,l-1])>p$.

\begin{fact}\label{fact:rect}
Let $R=[a,b]\times[c,d]$ be a set of columns times rows and let $\delta>0$.
There is a quantum algorithm that uses $O(\sqrt{|R|}\log(1/\delta))$ queries
and, with probability at least $1-\delta$, either correctly reports that $R$
contains no back edge, or returns the smallest column and the largest row
among the back edges in $R$.
\end{fact}
\begin{proof}
Apply Dürr--Høyer minimum finding~\cite{DH96} to the column indices of back edges, and then to the negated row indices. Each value query costs at most one query to $b(i,j)$. Repeating each search
$O(\log(1/\delta))$ times gives output with probability at least $1-\delta$.
\end{proof}

The recursive procedure $\textsc{Solve}(l,r,U,M)$ looks at the cuts
$l,\dots,r$ and is given two numbers:
\begin{itemize}
  \item the \emph{reach} $M$, the farthest row reached by back edges that start
        before column $l$; every cut $p$ with $l\le p<M$
        is already covered;
  \item the \emph{ceiling} $U$: every back edge that starts in columns $l,\dots,r$ ends at row $U$ or below, so searches never need to look higher.
\end{itemize}
It either reports $\textsc{uncovered}(p)$, where $p$ is the first uncovered cut
in the range, or reports $\textsc{covered}(M')$, meaning that every cut in the
range is covered and $M'$ is the new reach, now including the back edges that
start in columns $l,\dots,r$. It proceeds as shown in the Algorithm \ref{alg:solve}.\\

\noindent Further, to compute $\TC(T)$, we find a Hamiltonian path
(Lemma~\ref{lem:hampath}) and call $\textsc{Solve}(1,N-1,N,-\infty)$. If it
returns $\textsc{uncovered}(t)$, we output $\{v_1,\dots,v_t\}$; otherwise we
output $V$.

\begin{algorithm}[t]
\caption{$\textsc{TopCycle}(T)$}
\label{alg:topcycle}
\begin{algorithmic}[1]
\State find a Hamiltonian path $v_1\to v_2\to\dots\to v_N$
       \Comment{Lemma~\ref{lem:hampath}}
\State $M\gets-\infty$ \Comment{no back edge starts before column $1$}
\State $U\gets N$ \Comment{no back edge ends beyond row $N$}
\State $\mathit{ans}\gets\textsc{Solve}(1,N-1,U,M)$
\If{$\mathit{ans}=\textsc{uncovered}(t)$}
  \State \Return $\{v_1,\dots,v_t\}$
\Else
  \State \Return $V$ \Comment{every cut is covered}
\EndIf
\end{algorithmic}
\end{algorithm}

\begin{algorithm}[H]
\caption{$\textsc{Solve}(l,r,U,M)$}
\label{alg:solve}
\begin{algorithmic}[1]
\If{$l>r$} \Comment{empty range: nothing to check}
  \State \Return $\textsc{covered}(M)$
\EndIf
\If{$M>r$} \Comment{all cuts $l,\dots,r$ are already covered}
  \State Find the largest row $H$ of a back edge
         in $[l,r]\times[l+2,U]$ ($H=-\infty$ if none)\Comment{use Fact~\ref{fact:rect}}
  \State \Return $\textsc{covered}(\max(M,H))$ \Comment{only update the reach}
\EndIf
\If{$l=r$} \Comment{a single cut}
  \State  Find the largest row $H$ of a back edge
         in $[l,l]\times[l+2,U]$ ($H=-\infty$ if none)
  \If{$M>l$ \textbf{or} $H\neq-\infty$} \Comment{covered by an earlier edge
         or one in column $l$}
    \State \Return $\textsc{covered}(\max(M,H))$
  \Else
    \State \Return $\textsc{uncovered}(l)$ \Comment{cut $l$ is uncovered}
  \EndIf
\EndIf
\State $m\gets\lfloor(l+r)/2\rfloor$ \Comment{split at the middle cut}
\State Use Fact~\ref{fact:rect} on the rectangle $R=[l,m]\times[m+2,U]$ of
       leaving back edges
\If{$R$ contains a back edge}
  \State let $e$ be the smallest column and $H$ the largest row of a back edge
         in $R$
  \State \Comment{cuts $e,\dots,m$ are covered by the back edge in column $e$}
  \State $\mathit{ans}\gets\textsc{Solve}(l,e-1,m+1,M)$
         \Comment{cuts before $e$}
  \If{$\mathit{ans}=\textsc{uncovered}(p)$}
    \State \Return $\textsc{uncovered}(p)$ \Comment{first uncovered cut found}
  \EndIf
  \State let $\mathit{ans}=\textsc{covered}(M')$; set $M_R\gets\max(M',H)$ where $M'$ is the
        largest row of a back edge that starts in a column at most $r$
         \Comment{new reach}
\Else \Comment{no leaving back edge}
  \State $\mathit{ans}\gets\textsc{Solve}(l,m,m+1,M)$
         \Comment{the whole left half}
  \If{$\mathit{ans}=\textsc{uncovered}(p)$}
    \State \Return $\textsc{uncovered}(p)$ \Comment{first uncovered cut found}
  \EndIf
  \State let $\mathit{ans}=\textsc{covered}(M')$; set $M_R\gets M'$
         \Comment{new reach}
\EndIf
\State \Return $\textsc{Solve}(m+1,r,U,M_R)$ \Comment{right half, with the new
       reach}
\end{algorithmic}
\end{algorithm}

\begin{lemma}[Correctness of the Algorithm]\label{lem:solve-correct}
Assume that every search succeeds. Consider a call
$\textsc{Solve}(l,r,U,M)$ with $l\le r+1$ and $U\ge r+1$ such that
\begin{enumerate}[label=(I\arabic*)]
  \item every cut $p<l$ is covered;
  \item $M=\mathrm{reach}([1,l-1])$;
  \item every back edge $(i,j)$ with $l\le i\le r$ satisfies $j\le U$.
\end{enumerate}
If some cut in $[l,r]$ is uncovered, the call returns
$\textsc{uncovered}(p)$ for the smallest such cut $p$. Otherwise it returns
$\textsc{covered}(\mathrm{reach}([1,r]))$.
\end{lemma}
\begin{proof}
We prove this by induction on the width $w=r-l+1$. Note that,
by (I1), the first uncovered cut in $[l,r]$, if there is one, is the first
uncovered cut overall.

\emph{Steps (1)--(3): no recursion.}
If $l>r$ (Step~1), there is nothing to check, and
$\mathrm{reach}([1,r])=\mathrm{reach}([1,l-1])=M$ by (I2).
If $M>r$ (Step~2), the back edge that ends at row $M$ starts before column $l$
and ends after cut $r$, so it covers every cut in $[l,r]$; by (I3), the search
sees every back edge that starts in $[l,r]$, so
$\max(M,H)=\mathrm{reach}([1,r])$.
If $l=r$ (Step~3), cut $l$ is covered either by an earlier back edge, which
by~\eqref{eq:reach} happens exactly when $M>l$, or by a back edge that starts
in column $l$, which the search finds by (I3).

\emph{Step (4): splitting at $m$.}
By (I3), the rectangle $R$ contains exactly the leaving back edges: those that
start in $[l,m]$ and end at row $m+2$ or later. We check that the recursive
calls are given correct information and that their answers combine correctly.

\emph{The left call.} It has the same $l$ and $M$, so (I1) and (I2) still
hold. It also satisfies (I3) with ceiling $m+1$: in case~(a), a back edge that
starts before column $e$ and ends at row $m+2$ or later would be a leaving
back edge in a column smaller than $e$, which is impossible; in case~(b), there
are no leaving back edges at all. By induction, the left call either returns
the first uncovered cut of its range, which we return, or reports that all its
cuts are covered and returns their reach $M'$.

\emph{The cuts $e,\dots,m$ in case (a).} They are covered by the leaving back
edge that starts in column $e$, since it ends at row $m+2$ or later.

\emph{The new reach.} In case~(b), $M_R=M'=\mathrm{reach}([1,m])$. In
case~(a), the back edges that start before column $e$ reach $M'$, the leaving
back edges reach $H$, and every other back edge that starts in $[e,m]$ ends at
row at most $m+1<H$. Hence $M_R=\max(M',H)=\mathrm{reach}([1,m])$.

\emph{The right call.} At this point every cut up to $m$ is covered and
$M_R=\mathrm{reach}([1,m])$, so the right call satisfies (I1) and (I2), and
(I3) holds with the same $U$. By induction, it returns the correct answer for
$[m+1,r]$, which is then the correct answer for $[l,r]$.

The induction is well founded, because both recursive calls are narrower than
$[l,r]$: their widths are at most $\lceil w/2\rceil$ and $\lfloor w/2\rfloor$,
both less than $w$ when $w\ge2$.
\end{proof}

\noindent The top call $\textsc{Solve}(1,N-1,N,-\infty)$ satisfies
(I1)--(I3) trivially, so together with Lemma~\ref{lem:prefix} the output is
$\TC(T)$.

\emph{Total queries to compute $\textsc{Solve}(l,r,U,M)$:} Each call performs at most one search, over a rectangle of area at most
$wD$, where $w=r-l+1$ be its width and
$D=U-l+1$ its height. For a call with $w\ge2$, the left child has width at most
$\lceil w/2\rceil$ and height $(m+1)-l+1=\lceil w/2\rceil+1$, and the right
child has width $\lfloor w/2\rfloor$ and height
$U-m=D-\lceil w/2\rceil\le D$. Every call satisfies $w\le D$: this holds
for the top call ($N-1\le N$), for left children, and is preserved by right
children.

\emph{Depth and calls.} Widths shrink by a factor of at least $2/3$ from a
call to its children, so the recursion depth is $O(\log N)$. The column
intervals $[l,r]$ of calls at the same depth are disjoint, so their widths
sum to at most $N$. Every call has at most two children, so there are
$O(N\log N)$ calls.

Group the calls into chains $h_0,h_1,h_2,\dots$, where $h_0$
is the top call or a left child, and $h_{k+1}$ is the right child of $h_k$.
Every call lies in exactly one chain. Along a chain, $w(h_k)\le w(h_0)/2^k$
and $D(h_k)\le D(h_0)$, so, using $w(h_0)\le D(h_0)$,
\[
  \sum_{k\ge0}\sqrt{w(h_k)D(h_k)}
  \le \sum_{k\ge0} \frac{D(h_0)}{2^{k/2}}
  = (2+\sqrt2)\,D(h_0).
\]

The top call has $D=N$. A left child of a call $q$ with
$w(q)\ge2$ has
$D=\lceil w(q)/2\rceil+1\le w(q)/2+3/2\le 2\,w(q)$, and each call has at
most one left child. Summing over depths gives
\[
  \sum_{\text{chains}} D(h_0)\le N+\sum_{q}2\,w(q)
  = N+O(N\log N)=O(N\log N). \qedhere
\]
\begin{theorem}\label{thm:topcycle-upper}
The top cycle of an $N$-vertex tournament can be computed with bounded error
using $O(N\log^2 N)$ quantum queries.
\end{theorem}

\begin{proof}
Using Algorithm \ref{alg:topcycle} with
Fevery search performed by Fact~\ref{fact:rect} with $\delta=1/N^3$. A search
over an area at most $wD$ costs $O(\sqrt{wD}\log N)$ queries, so the total query is $O(N\log^2 N)$.
\end{proof}

\begin{theorem}\label{thm:topcycle-lb}
Every bounded-error quantum algorithm that computes the top cycle of
$N$-vertex tournaments makes $\Omega(N)$ queries.
\end{theorem}

\begin{proof}
We follow the construction of~\cite[Theorem~3.8]{MD24}. Let
$n=\lfloor N/2\rfloor\ge3$, and take vertex sets
$A=\{a_1,\dots,a_n\}$ and $B=\{b_1,\dots,b_n\}$, plus one vertex $z$ that
loses to every other vertex if $N$ is odd. Inside $A$, orient
$a_i\to a_{i+1\bmod n}$ for every $i$ and the remaining pairs arbitrarily
but fixed; do the same inside $B$. Then $A$ is strongly connected. Given
$y\in\{0,1\}^{n\times n}$, orient $b_j\to a_i$ if $y_{ij}=1$ and
$a_i\to b_j$ otherwise. Each tournament query costs at most one query to
$y$.

If $y=0$, then $A$ beats every vertex outside $A$, so $A$ is dominant; and
no proper nonempty subset of $A$ is dominant, since $A$ is strongly
connected. As dominant sets are ordered by inclusion, $\TC(T)=A$. If
$y\neq0$, then some $b_j$ beats some $a_i$, so $A$ is not dominant and
$\TC(T)\neq A$. Hence an algorithm computing $\TC(T)$ computes
$\mathrm{OR}_{n^2}(y)$, which requires $\Omega(\sqrt{n^2})=\Omega(N)$
queries (Theorem~\ref{thm:or}).
\end{proof}

\section{Uncovered Set}\label{sec:uncovered}

In a tournament, the uncovered set is exactly the set of kings. A single king
can be found with $\tilde O(\sqrt N)$ quantum queries~\cite{MPS23}, and the top
cycle, which contains all kings, can be computed with $\tilde O(N)$ queries
(Section~\ref{sec:topcycle}). We now show that computing the \emph{whole}
uncovered set is harder: it requires $\Omega(N^{3/2})$ queries.

To prove the lower bound we will define a problem as follows:

\begin{definition}[Tight and deficient matrices]\label{def:tight}
Let $W=(w_{ij})$ be an $s\times s$ Boolean matrix. We say that $W$ is
\begin{itemize}
  \item \emph{tight:} if every column of $W$ contains exactly one $1$;
  \item \emph{deficient:} if exactly one column of $W$ is all $0$, and every
        other column contains exactly one $1$.
\end{itemize}
\end{definition}

\noindent Changing a $0$ in the all-zero column of a deficient matrix into a
$1$ makes it tight, and changing any $1$ of a tight matrix into a $0$ makes it
deficient.

\begin{definition}[\textsc{Tight-Count}$_{k,c}$]\label{def:tightcount}
For $0\le c\le k$, let $\mathcal P_c$ be the set of $k$-tuples
$(W_1,\dots,W_k)$ of $s\times s$ Boolean matrices in which every $W_t$ is
tight or deficient and exactly $c$ of them are tight. In the problem
\textsc{Tight-Count}$_{k,c}$ we are given an input from
$\mathcal P_c\cup\mathcal P_{c+1}$ and must decide which of the two sets
contains it.
\end{definition}

\begin{lemma}\label{lem:puzzles-lb}
Let $0\le c\le k-1$. Every bounded-error quantum algorithm for
\textsc{Tight-Count}$_{k,c}$, that is, every algorithm that distinguishes
$\mathcal P_c$ from $\mathcal P_{c+1}$, makes
$\Omega\bigl(s\sqrt{(k-c)(c+1)}\bigr)$ queries to $y$.
\end{lemma}

\begin{proof}
We apply Theorem~\ref{thm:ambainis} with $X=\mathcal P_c$ and
$Y=\mathcal P_{c+1}$, relating $w\in X$ and $w'\in Y$ if $w'$ is obtained
from $w$ by changing a single $0$ into a $1$.

\emph{Parameter $M$.} Let $w\in X$. Changing a $0$ in the all-zero column of
one of the $k-c$ deficient blocks gives an element of $Y$, and there are $s$
such cells in each of these blocks. Changing any other $0$ creates a column
with two $1$s, which is not allowed. Hence $M=(k-c)s$.

\emph{Parameter $M'$.} Let $w'\in Y$. Changing one of the $s$ ones of one of
the $c+1$ tight blocks into a $0$ gives an element of $X$. Changing a $1$ of a
deficient block would create a second all-zero column, which is not allowed.
Hence $M'=(c+1)s$.

\emph{Parameters $L$ and $L'$.} Related inputs differ in exactly one position,
so $L=L'=1$.

Thus, theorem~\ref{thm:ambainis} gives the bound
$\sqrt{MM'/(LL')}=s\sqrt{(k-c)(c+1)}$.
\end{proof}

We will reduce the \textsc{Tight-Count}$_{k,c}$ problem to the Uncovered set problem. Let $y$ be the instance of \textsc{Tight-Count}$_{k,c}$.
\paragraph{Construction:}
Let $g,s\ge1$. The tournament $T_y$ has vertex set
$V(T_y)=A\cup B\cup X\cup Z\cup\{h\}$, where
$A=\bigcup_{p=1}^g A_p$ and $B=\bigcup_{q=1}^g B_q$. Each block
$A_p$ and $B_q$ contains $s$ vertices, while
$X=\{x_{p,q}:p,q\in[g]\}$ and $Z=\{z_{p,q}:p,q\in[g]\}$.
Thus, $T_y$ has $N=2gs+2g^2+1$ vertices.

The input is a string
$y\in\{0,1\}^{A\times B}$, and for $a\in A$, $b\in B$ we let $a\to b$ if
$y_{ab}=1$ and $b\to a$ otherwise. All other edges are fixed:
\begin{enumerate}[label=(E\arabic*)]
  \item $x_{p,q}$ beats every vertex of $A_p$, loses to every vertex of
        $A\setminus A_p$, and loses to every vertex of $B$;
  \item $x_{p,q}\to z_{p,q}$, and $z_{p',q'}\to x_{p,q}$ whenever
        $(p',q')\neq(p,q)$;
  \item $x_{p,q}\to h$ for all $p,q$;
  \item $z_{p,q}$ beats every vertex of $A$ and of $B\setminus B_q$, and loses
        to every vertex of $B_q$;
  \item $h$ beats every vertex of $A$ and every helper $z_{p,q}$, and loses to
        every vertex of $B$;
  \item all remaining pairs (inside $A$, inside $B$, among the puzzle
        vertices, and among the helpers $z_{p,q}$) are oriented arbitrarily
        but fixed, independently of $y$.
\end{enumerate}

\begin{figure}[t]
\centering
\begin{tikzpicture}[
  >=Stealth,
  every node/.style={font=\small},
  blk/.style={draw, rounded corners, minimum width=1.3cm, minimum height=7mm},
  grp/.style={draw, rounded corners, minimum width=3.2cm, minimum height=9mm,
              align=center},
  lab/.style={font=\scriptsize, inner sep=1.5pt}]
  % the sets A and B
  \node[blk] (A1) at (-5.4, 2.3) {$A_1$};
  \node[blk] (A2) at (-5.4, 1.1) {$A_2$};
  \node      (Ad) at (-5.4,-0.3) {$\vdots$};
  \node[blk] (Ag) at (-5.4,-1.7) {$A_g$};
  \node[draw, rounded corners, fit=(A1)(Ag), inner sep=8pt,
        label=below:{$A$}] (A) {};
  \node[blk] (B1) at (5.4, 2.3) {$B_1$};
  \node[blk] (B2) at (5.4, 1.1) {$B_2$};
  \node      (Bd) at (5.4,-0.3) {$\vdots$};
  \node[blk] (Bg) at (5.4,-1.7) {$B_g$};
  \node[draw, rounded corners, fit=(B1)(Bg), inner sep=8pt,
        label=below:{$B$}] (B) {};
  % the middle column
  \node[circle, draw, minimum size=8mm] (h) at (0, 2.6) {$h$};
  \node[grp] (X) at (0, 0.5) {$x_{p,q}$};
  \node[grp] (Z) at (0,-2.0) {$z_{p,q}$};
  % input edges
  \draw[dashed] (A.north) to[bend left=22]
        node[above, lab]{input edges: $a\to b$ iff $y_{ab}=1$} (B.north);
  % (E1)
  \draw[->] ([yshift=2mm]X.west) -- ([yshift=2mm]X.west -| A.east)
        node[midway, above, lab]{$x_{p,q}\to A_p$};
  \draw[->] ([yshift=-2mm]X.west -| A.east) -- ([yshift=-2mm]X.west)
        node[midway, below, lab]{$A\setminus A_p\to x_{p,q}$};
  \draw[->] (X.east -| B.west) -- (X.east)
        node[pos=0.25, above, lab]{all};
  % (E2)
  \draw[->] ([xshift=-4mm]X.south) -- ([xshift=-4mm]X.south |- Z.north)
        node[midway, left, lab]{$x_{p,q}\to z_{p,q}$};
  \draw[->] ([xshift=4mm]Z.north) -- ([xshift=4mm]Z.north |- X.south)
        node[midway, right, lab]{$z_{p',q'}\to x_{p,q}$};
  % (E3)
  \draw[->] (X.north) -- (h.south) node[midway, right, lab]{all};
  % (E4)
  \draw[->] (Z.west) -- (Z.west -| A.east) node[midway, above, lab]{all};
  \draw[->] ([yshift=2mm]Z.east) -- ([yshift=2mm]Z.east -| B.west)
        node[midway, above, lab]{$z_{p,q}\to B\setminus B_q$};
  \draw[->] ([yshift=-2mm]Z.east -| B.west) -- ([yshift=-2mm]Z.east)
        node[midway, below, lab]{$B_q\to z_{p,q}$};
  % (E5)
  \draw[->] (h) -- (h -| A.east) node[midway, above, lab]{all};
  \draw[->] (h -| B.west) -- (h) node[midway, above, lab]{all};
  \draw[->] (h.south east) .. controls (3.8,2.0) and (3.8,-1.0)
        .. (Z.north east) node[pos=0.22, above right, lab]{all};
\end{tikzpicture}
\caption{Construction of the tournament $T_y$}
\label{fig:uncovered-construction}
\end{figure}
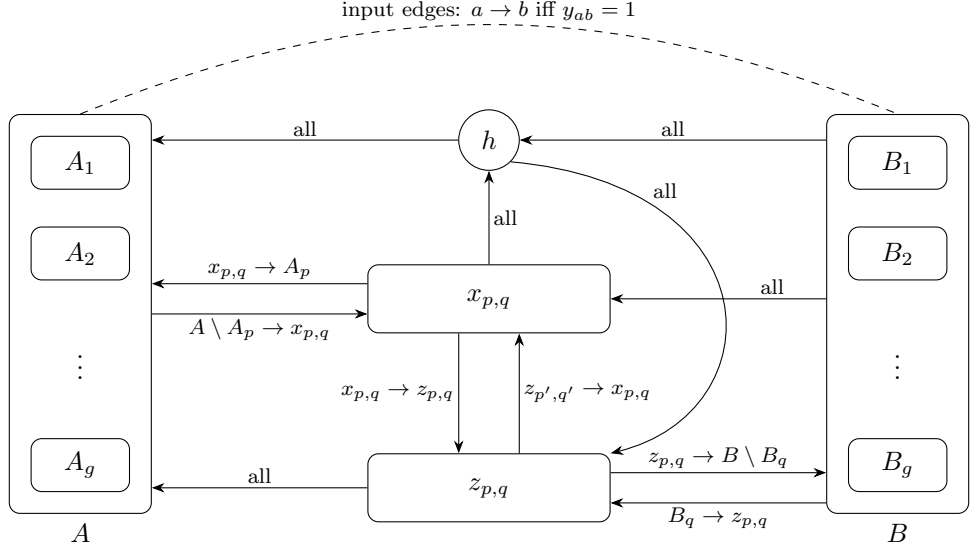
Only the edges between $A$ and $B$ depend on $y$, and each of them depends on a
single bit, so every query to $T_y$ costs at most one query to $y$.

For $p,q\in[g]$, the \emph{block} $(p,q)$ is the $s\times s$ Boolean matrix
$W_{p,q}=(y_{ab})_{a\in A_p,\,b\in B_q}$, with rows indexed by $A_p$ and
columns indexed by $B_q$. Every bit of $y$ belongs to exactly one block, so
$y$ is the same as the $k$-tuple $(W_{p,q})_{p,q\in[g]}$ of $k=g^2$ matrices.
The column of $W_{p,q}$ indexed by $b\in B_q$ contains a $1$ if and only if
$b$ is beaten by some vertex of $A_p$.

\begin{lemma}\label{lem:puzzle}
For all $p,q\in[g]$, the vertex $x_{p,q}$ is a king of $T_y$ if and only
if $W_{p,q}$ is tight.
\end{lemma}

\begin{proof}
Write $x=x_{p,q}$. By (E1)--(E3), the out-neighbours of $x$ are the vertices
of $A_p$, $z_{p,q}$ and $h$, and possibly some other $x_{p',q'}$ vertices. We check,
for each vertex $u$ with $u\to x$, whether $x$ reaches $u$ in two steps.
\begin{itemize}
  \item If $u\in A\setminus A_p$, $u\in B\setminus B_q$, or $u$ is another $x_{p',q'}$ vertex, then $x\to z_{p,q}\to u$ by (E2) and (E4).
  \item If $u=z_{p',q'}$ with $(p',q')\neq(p,q)$, then $x\to h\to u$ by (E3)
        and (E5).
  \item If $u=b\in B_q$, then $b$ beats $z_{p,q}$ by (E4), beats $h$ by (E5),
        and beats every $X$ vertices by (E1). Hence $x$ reaches $b$ in two
        steps if and only if some $a\in A_p$ satisfies $a\to b$, that is,
        $y_{ab}=1$.
\end{itemize}
Therefore $x$ is a king if and only if every $b\in B_q$ is beaten by some
$a\in A_p$, which means that every column of $W_{p,q}$ contains a $1$. A tight
matrix has this property, and a deficient matrix does not.
\end{proof}

\begin{theorem}\label{thm:uncovered-lb}
Every bounded-error quantum algorithm that computes the uncovered set of
$N$-vertex tournaments makes $\Omega(N^{3/2})$ queries.
\end{theorem}
\begin{proof}

Let $N=4g^2+1$, and set $s=g$, $k=g^2$, and
$c=\lfloor k/2\rfloor$ in the construction. Let
$y\in\mathcal P_c\cup\mathcal P_{c+1}$. By Lemma~\ref{lem:puzzle},
$x_{p,q}\in\mathrm{UC}(T_y)$ if and only if $W_{p,q}$ is tight.
Consequently, $\mathrm{UC}(T_y)$ contains exactly $c$ vertices of the
form $x_{p,q}$ when $y\in\mathcal P_c$, and exactly $c+1$ when
$y\in\mathcal P_{c+1}$. Since the collection of vertices $x_{p,q}$ is
fixed independently of $y$, an algorithm that outputs
$\mathrm{UC}(T_y)$ can count these vertices without additional queries.
Thus, it solves $\textsc{Tight-Count}_{k,c}$. Each query to $T_y$ can
be simulated using at most one query to $y$, so Lemma~\ref{lem:puzzles-lb}
gives a lower bound of
$\Omega\bigl(s\sqrt{(k-c)(c+1)}\bigr)=\Omega(g^3)
=\Omega(N^{3/2})$, since $s=g$ and both $k-c$ and $c+1$ are at least
$k/2$.\\

Now consider an arbitrary $N\ge5$. Let $g$ be the largest integer such
that $N_0=4g^2+1\le N$; then $g=\Theta(\sqrt N)$ and $N_0=\Theta(N)$.
Extend $T_y$ to a tournament $T'_y$ on $N$ vertices by adding
$N-N_0$ new vertices, orienting the edges among them transitively, and
making every new vertex lose to every vertex of $T_y$.

We have $\mathrm{UC}(T'_y)=\mathrm{UC}(T_y)$. Every king of $T_y$
remains a king, since it beats every new vertex. A non-king of $T_y$
remains a non-king: a new vertex cannot provide a two-edge path to a
previously unreachable vertex of $T_y$, since every new vertex loses to
every vertex of $T_y$. Finally, no new vertex is a king. It has no edge
to a vertex of $T_y$, and any two-edge path from it to $T_y$ would have
to pass through another new vertex, which also loses to every vertex of
$T_y$.

The added edges are independent of $y$, so each query to $T'_y$ can
still be simulated using at most one query to $y$. Hence an algorithm
for computing the uncovered set of an $N$-vertex tournament would solve
$\textsc{Tight-Count}_{k,c}$ with the same query complexity. The first
part therefore gives a lower bound of
$\Omega(g^3)=\Omega(N_0^{3/2})=\Omega(N^{3/2})$.
\end{proof}

\section*{Acknowledgments}
We gratefully acknowledge Bireswar Das for his valuable insights,
constructive comments, and fruitful discussions throughout this project.

Udit Kumar is partially supported by the Anusandhan National Research Foundation (ANRF),
India, through the project ``Quantum Algorithms for Group Theoretic
Problems.''

Parts of this work were developed with assistance from GPT-6 Astra and Claude
Opus~5.5. In particular, these tools suggested the \prob{Heavy-Row} problem and
the use of Ambainis's adversary method to prove its lower bound, and
contributed ideas to our algorithm for the top cycle. The authors verified all
results and proofs and take full responsibility for the content of this paper.
\bibliography{References}
\end{document}